\documentclass[a4paper]{article}
\usepackage{amsthm,amssymb,amsmath,enumerate,graphicx,psfrag}
\usepackage{mdwlist}
\usepackage{tikz}

\newtheorem{definition}{Definition}

\newtheorem{theorem}[definition]{Theorem}

\newtheorem{lemma}[definition]{Lemma}

\newcommand{\selfcite}[1]{$\!\!${\bf\cite{#1}}}
\renewcommand{\epsilon}{\varepsilon}
\newcommand{\ide}[2]{/_{#1=#2}}

\newcommand{\comment}[1]{}

\newcommand{\emtext}[1]{\text{\em #1}}

\makeatletter
\newcommand{\saferef}[1]{%
    \@ifundefined{r@#1}{0}{\ref{#1}}
}
\makeatother
\usepackage{pifont}
\newcounter{algcnt}
\newcommand{\cnumber}[1]{\setcounter{algcnt}{191}\addtocounter{algcnt}{#1}{\large\ding{\value{algcnt}}}}
\newcommand{\mylbl}{\setcounter{algcnt}{191}\addtocounter{algcnt}{\value{enumi}}{\large\ding{\value{algcnt}}}}
\newcommand{\aref}[1]{\cnumber{\saferef{#1}}} 
\newcommand{\aresume}{\resume{enumerate}\renewcommand{\labelenumi}{{\mylbl}}\sffamily}
\newcommand{\astart}{\begin{enumerate}\renewcommand{\labelenumi}{{\mylbl}}\sffamily}

\newcommand{\returnvalue}[1]{\emph{``#1''}}

\newcommand{\sm}{\setminus}

\title{Claw-free $t$-perfect graphs can be recognised in polynomial time}
\author{Henning Bruhn and Oliver Schaudt}

\begin{document}
\maketitle

\begin{abstract}
A graph is called $t$-perfect if its stable set polytope is defined by non-negativity,
edge and odd-cycle inequalities. 
We show that it can be decided in polynomial time whether a given claw-free graph is $t$-perfect.
\end{abstract}

\section{Introduction}

We treat $t$-perfect graphs, a class of graphs that is not only 
similar in name to perfect graphs but also shares a number of 
their properties. One way to define perfect graphs is via the 
stable set polytope: The convex hull of all characteristic vectors
of stable sets (sets of pairwise non-adjacent vertices). 
As shown independently by Chv\'atal~\cite{Chvatal75} and Padberg~\cite{Pad74},
a graph 
 is perfect if and only if its stable set polytope 
is determined by non-negativity and clique inequalities. 
In analogy, Chv\'atal~\cite{Chvatal75} proposed to study the class of graphs 
whose stable set polytope is defined by non-negativity, edge 
and odd-cycle inequalities. These graphs became to be known 
as \emph{$t$-perfect graphs}. (We defer precise and more explicit definitions
to the next section.) 

Two celebrated results on perfect graphs are the proof of the 
strong perfect graph conjecture by Chudnovsky, Robertson, Seymour 
and Thomas~\cite{SPGT06} and the polynomial time algorithm of 
Chudnovsky, Cornu\'ejols, Liu, Seymour and Vu\v{s}kovi\'c~\cite{CCLSV05}
that checks whether a given graph is perfect or not. Analogous
results for $t$-perfection seem desirable but out of reach 
for the moment. Restricted to claw-free graphs, however, this 
changes. A characterisation of claw-free $t$-perfect graphs
in terms of forbidden substructures was recently proved by
Bruhn and Stein~\cite{BS12}.
In this work we present a
recognition algorithm for $t$-perfect claw-free graphs:

\begin{theorem}\label{thm:detect}
It can be decided in polynomial time whether a given claw-free graph is $t$-perfect.
\end{theorem}

The class of $t$-perfect graphs seems rich and of non-trivial structure. 
Examples include series-parallel graphs (Boulala and Uhry~\cite{BouUhr79})
and bipartite or almost bipartite graphs. More classes were identified by
Shepherd~\cite{Shepherd95} and 
Gerards and Shepherd~\cite{GerShe98}. 
An attractive result on the algorithmic side is the  combinatorial polynomial-time algorithm  
of
Eisenbrand, Funke, Garg and K\"onemann~\cite{EFGK02}
that solves the max-weight stable 
set problem on $t$-perfect graphs. 

There is also an, at least superficially, more stringent notion of $t$-perfection, 
\emph{strong $t$-perfection}; see  Schrijver~\cite[Vol.~B, Ch.~68]{LexBible}
where also some background on $t$-perfect graphs may be found.
Interestingly, there is no $t$-perfect graph known that fails 
to be strongly $t$-perfect. In fact, for some classes these two notions
are known to be equivalent, see 
Schrijver~\cite{Schrijver02} and
Bruhn and Stein~\cite{strongtp}.

The graphs whose stable set polytope is given by non-negativity, clique 
and odd-cycle inequalities are called \emph{$h$-perfect}. The class
of $h$-perfect graphs is a natural superclass of both perfect as well 
as $t$-perfect graphs. The class has been studied by
Fonlupt and Uhry~\cite{FonUhr82},
Sbihi and Uhry~\cite{SbiUhr84}, and Kir\'aly and P\'ap~\cite{KP07,KP08}.

\medskip

We briefly outline the strategy of our recognition algorithm.
In Sections~\ref{sec:linegraphs} and~\ref{sec:nomorelinkages}, 
we show how to recognise $t$-perfect line graphs.
For this, we work in the underlying source graph that gives rise
to the line graph. In the source graph we need to 
detect certain subgraphs called \emph{thetas}: 
two vertices joined by three disjoint paths. In the thetas that are of interest to us
the linking paths have to respect additional parity constraints. 

The general algorithm for claw-free graphs is presented in Sections~\ref{clawsec} and~\ref{sec:ClawLemma}
and relies on a 
divide and conquer approach to split the input graph along small separators.
In this phase of the algorithm, we make extensive use 
of a procedure by van 't Hof, Kami\'nski and Paulusma~\cite{HKP12} 
that detects induced paths of given parity in claw-free graphs.
The final pieces that cannot be split anymore 
turn out to be essentially 
line graphs, which we already dealt with.

\section{Claw-free graphs and $t$-perfection}

We refer to Diestel~\cite{diestelBook10} for general notation and definitions concerning graphs. 

Let us recall the definition of a claw-free graph.
The \emph{claw} is the graph $G=(V,E)$ with $V=\{u,v_1,v_2,v_3\}$ and $E=\{uv_1,uv_2,uv_3\}$,
and we call $u$ its \emph{centre}.
A graph is called \emph{claw-free} if it does not contain 
an induced subgraph that is isomorphic to the claw.
Claw-free graphs form a superclass of line graphs.

In order to define $t$-perfection, we associate with 
every graph $G=(V,E)$ 
a polytope denoted TSTAB$(G)$,  the set
of all vectors $x\in\mathbb R^{V}$ satisfying 
\begin{eqnarray}
0\leq x_v\leq 1 &&\text{ for every vertex }v\in V,\notag\\
\label{tstab}x_u+x_v\leq 1 &&\text{ for every edge }uv\in E,\\
\sum_{v\in V(C)}x_v\leq \lfloor \tfrac{1}{2}|V(C)|\rfloor &&\text{ for every odd cycle }C
\text{ in }G.\notag
\end{eqnarray}
The graph $G$ is called \emph{$t$-perfect} if TSTAB$(G)$
coincides with the stable set polytope of $G$ (the convex hull
of characteristic vectors of stable sets in $\mathbb R^{V}$).
An alternative but equivalent definition is to say that $G$ is $t$-perfect
if and only if TSTAB$(G)$ is an integral polytope. 
\medskip

As observed by Gerards and Shepherd~\cite{GerShe98}, 
the following operation called \emph{$t$-contraction} preserves $t$-perfection: 
Contraction of all edges incident with any vertex $v$ whose
neighbourhood $N(v)$ is a stable set. We then say that a  \emph{$t$-contraction
is performed at $v$}.
If $G$ is claw-free, the $t$-contraction becomes particularly simple. Indeed, 
a $t$-contraction at $v$ is only possible if $v$ has degree~$\leq 2$; otherwise
$v$ is the centre of a claw. If $v$ has precisely two neighbours $u$ and $w$
then the $t$-contraction simply identifies $u,v,w$ to a single vertex.

To characterise the class of $t$-perfect graphs in terms of forbidden substructures, 
the concept of $t$-minors was introduced in~\cite{strongtp}:
A graph $H$ is a \emph{$t$-minor} of a graph $G$ 
if $H$ can be obtained from $G$ by a series of vertex deletions and/or $t$-contractions.
Note that the class of $t$-perfect graphs is closed under taking $t$-minors.

We note an easy but useful observation~\cite{strongtp}:
\begin{equation}\label{keepscf}
\emtext{
any $t$-minor of a claw-free graph is claw-free.
}
\end{equation}

It turns out that $t$-perfect claw-free graphs can be characterised in terms of finitely many forbidden $t$-minors:
\begin{theorem}[Bruhn and Stein~\cite{BS12}]
\label{thm:Char}
A claw-free graph is $t$-perfect if and only if it
does not contain any of $K_4$, $W_5$, $C^2_7$ and $C^2_{10}$ as a $t$-minor.
\end{theorem}
Here, $K_4$ denotes the complete graph on four vertices, $W_5$ 
is the $5$-wheel, and for $n\in\mathbb N$ we denote by $C^2_n$
the square of the cycle $C_n$ on $n$ vertices, see Figure~\ref{fig:culprits}.
More precisely, we define $C^2_n$ always on the vertex set $v_1,\ldots,v_n$,
so that $v_i$ and $v_j$ are adjacent if and only if $|i-j|\leq 2$,
where we take the indices modulo~$n$.

      \begin{figure}[ht]
      \centering
      \includegraphics[scale=0.7]{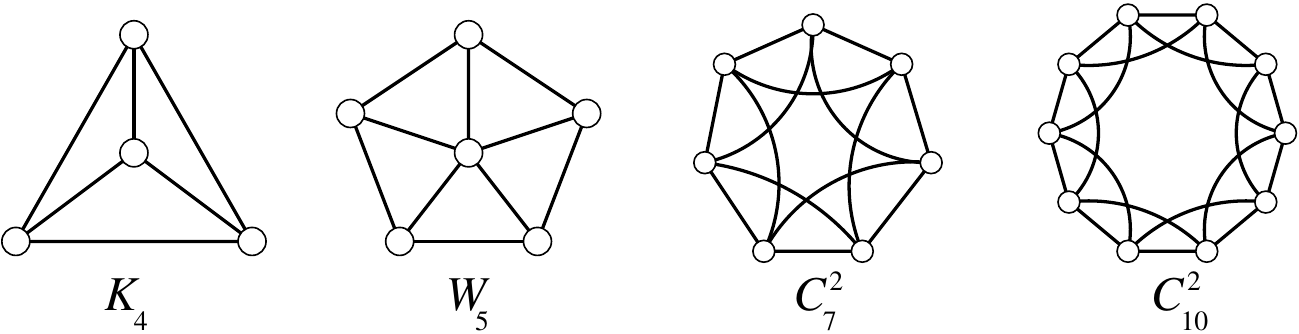}
      \caption{The forbidden $t$-minors.}\label{fig:culprits}
      \end{figure}

We often present our algorithms intermingled with parts of the corresponding
correctness proofs. To set the algorithm steps apart from the surrounding proofs
we write them as follows: 
\astart
\item The first line of an algorithm.
\end{enumerate}

Finally, for two vertices $u,v$, a \emph{$u$--$v$-path} is simply a path from $u$ to $v$.
Similarly, if $X,Y \subseteq V(G)$, 
then we mean by an 
\emph{$X$--$Y$-path}  a path from a vertex in $X$ to some vertex in $Y$ so that no 
internal vertex belongs to $X\cup Y$.
In the case that $X=Y$ we simply speak of an $X$-path.

\section{Line graphs}\label{sec:linegraphs}

We first solve the recognition problem for line graphs:
\begin{lemma}\label{detectline}
It can be decided in polynomial time whether 
the line graph of a given graph is $t$-perfect.
\end{lemma}
We develop the algorithm in the course of this section and the next.
That the algorithm is correct is based on the following characterisation 
of $t$-perfect line graphs. 

We call a graph \emph{subcubic} if its maximum degree is at most 3.
A \emph{skewed theta} is  a subgraph which is the union of three edge-disjoint paths linking
two vertices, called \emph{branch vertices}, such that two paths have odd length and one has even length.
Note that a skewed theta does not have to be an induced subgraph.

\begin{lemma}\selfcite{BS12}\label{tpline}
Let $G$ be a graph.
Then the line graph $L(G)$ is $t$-perfect 
if and only if $G$ is subcubic and does not contain
any skewed theta.
\end{lemma}

Checking for subdivisions of a certain graph can often be reduced to 
the well-known \textsc{$k$-Disjoint Paths} problem: Given a number of $k$ 
pairs of terminal vertices, the task is to decide whether there are disjoint
paths joining the paired terminals. In our context, however, this is not sufficient
as the paths linking the branch vertices in a skewed theta are subject to 
parity constraints. 

That this deep and seemingly hard problem, \textsc{$k$-Disjoint Paths with Parity Constraints},
allows nevertheless a polynomial time algorithm has been announced by 
Kawarabayashi, Reed and Wollan~\cite{KRW11}.
Another algorithm was given in the PhD thesis of Huynh~\cite{Huy09}.
These are very impressive results indeed, and they draw on deep insights coming from the 
graph minor project of Robertson and Seymour and its extension to matroids by
Geelen, Gerards and Whittle. For both algorithms, however, it seems doubtful
whether they could be implemented with a reasonable amount of work (or at all). 
We prefer therefore to present a more elementary algorithm for Lemma~\ref{detectline} 
that does not rely on any deep result and that is, in principle, implementable. 
\medskip

Given a bipartition $\mathcal P=(A,B)$ (where we allow $A$ or $B$ to be empty)
of the vertex set of a graph $G$, we call an edge \emph{$\mathcal P$-even}
if its endvertices lie in distinct partition classes of $\mathcal P$;
otherwise the edge is \emph{$\mathcal P$-odd.}
We observe that
a cycle is odd if and only if it contains an odd number of $\mathcal P$-odd edges.

%
%

The algorithm we present here to check for skewed thetas runs in two phases. 
We start with any bipartition~$\mathcal P$.
In the first phase, the algorithm tries to
iteratively reduce the number of $\mathcal P$-odd edges.
If this is no longer possible we either have found a skewed theta 
or we have arrived at a bipartition $\mathcal P'$ with at most two $\mathcal P'$-odd
edges. Then, in the second phase, we exploit that any skewed theta has to 
contain at least one of the at most two $\mathcal P'$-odd edges. In that case, 
it becomes possible to check directly for a skewed theta:

\begin{lemma}\label{lem:LinkagesAlgo}
Given a graph $G$ 
and a 
bipartition $\mathcal P$ of $V(G)$ so that at most two edges 
are $\mathcal P$-odd, it is possible
to check in polynomial time 
whether $G$ contains
a skewed theta.
\end{lemma}

The proof of Lemma~\ref{lem:LinkagesAlgo} is deferred to Section~\ref{sec:nomorelinkages}.
In the remainder of this section, we show how to iteratively reduce the number of $\mathcal P$-odd edges.
We start with two lemmas that give sufficient conditions for the existence of a skewed theta.

\begin{lemma}\label{2oddcycs}
A $2$-connected subcubic graph that contains two edge-disjoint odd cycles
contains a skewed theta.
\end{lemma}
\begin{proof}
Let $C_1$ and $C_2$ be two edge-disjoint odd cycles in $G$,
which then are also vertex-disjoint as the graph is assumed to be subcubic.
Since $G$ is $2$-connected there are two disjoint $C_1$--$C_2$-paths $P_1,P_2$.
The endvertices of $P_1$ and $P_2$ subdivide
$C_2$ into two subpaths, and
one of these subpaths together with $P_1$ and $P_2$ yields an odd
$C_1$-path, and thus a skewed theta. 
\end{proof}

For any bipartition $\mathcal P$ of $G$ define $G_\mathcal P$ 
to be the (bipartite) subgraph on $V(G)$ together with all the $\mathcal P$-even
edges. We formulate a second set of conditions that implies 
the presence of a skewed theta.

Let $C$ be a cycle and let $P$ and $Q$ be two disjoint $C$-paths.
Let $p_1,p_2$ be the endpoints of $P$ and $q_1,q_2$ be the endpoints of $Q$.
We say that $P$ and $Q$
are \emph{crossing on $C$} if $p_1,q_1,p_2,q_2$ appear in this order on $C$.

\begin{lemma}\label{lem:ThreeOddEdgeCut}
Let $G$ be a subcubic graph with a bipartition $\mathcal P$.
Let there be three $\mathcal P$-odd edges $o_1,o_2,o_3$ and two disjoint trees $T_1,T_2\subseteq G_\mathcal P$,
each containing an endvertex of each of $o_1,o_2,o_3$.

Assume the trees are minimal subject to the above description.
If $G_\mathcal P$ contains three edge-disjoint $T_1$--$T_2$-paths
then $G$ contains a skewed theta.
\end{lemma}
\begin{proof}
Throughout the proof, we assume that $G$ does not contain a skewed theta.
Our aim is to show that $G_\mathcal P$ does not contain three edge-disjoint $T_1$--$T_2$-paths.

For this, we first prove a sequence of more general claims.
Let $r_1r_2$ and $s_1s_2$ be two $\mathcal P$-odd edges of $G$ such that
there are two disjoint paths 
$R_1=r_1\ldots s_1$, $R_2=r_2\ldots s_2$.
Let $C$ be the cycle $r_1R_1s_1s_2R_2r_2r_1$.


We claim that
\begin{equation}\label{2crossing}
\begin{minipage}[c]{0.8\textwidth}\em
any two edge-disjoint $R_1$--$R_2$-paths $P,Q$ are crossing on $C$.
\end{minipage}\ignorespacesafterend 
\end{equation} 
If $P$ and $Q$ are not crossing then we can easily find two edge-disjoint cycles
in $R_1\cup R_2\cup P\cup Q$, one through $r_1r_2$ and the other through $s_1s_2$.
By Lemma~\ref{2oddcycs}, however, this is impossible. Thus, $P$ and $Q$ are crossing.

Next, we show that
\begin{equation}\label{2colours}
\begin{minipage}[c]{0.8\textwidth}\em
the endvertices of any two edge-disjoint $R_1$--$R_2$-paths $P,Q$ 
in $R_1$ lie in distinct partitions classes of $\mathcal P$.
\end{minipage}\ignorespacesafterend 
\end{equation} 
Denote the endvertex of $P$ 
in $R_1$ by $p_1$ and denote the one in $R_2$ by $p_2$; define 
$q_1,q_2$ analogously for $Q$.

Suppose that $p_1$ and $q_1$ lie in the same partition class of $\mathcal P$.
Since $G$ is subcubic, $P$ and $Q$ are disjoint, and, by~\eqref{2crossing}, crossing. 
Assume that $p_1\in r_1R_1q_1$.
As $p_1$ and $q_1$ are contained in the same partition class,
 the path $p_1R_1q_1$ has even length. On the other hand, the following
two paths have odd length: $p_1Pp_2R_2s_2s_1R_1q_1$ and 
$q_1Qq_2R_2r_2r_1R_1p_1$. As, moreover, these three paths meet only in $p_1$
and $q_1$ we have found a skewed theta; this proves~\eqref{2colours}.

From this follows that
\begin{equation}\label{no3}
\emtext{
$G$ cannot contain three edge-disjoint $R_1$--$R_2$-paths.
}
\end{equation}
Indeed, by~\eqref{2colours}, the three endvertices of such paths in $R_1$
would need to lie in distinct partition classes, which is clearly impossible
as $\mathcal P$ is a bipartition.

To complete the proof, suppose now that $G_\mathcal P$ contains three edge-disjoint $T_1$--$T_2$-paths $P_1,P_2,P_3$.
Denote by $t_i$ the unique vertex that separates all the endvertices of $o_1,o_2,o_3$
in $T_i$ (unless $T_i$ is a path this is the vertex of degree~$3$ in $T_i$).
Observe that $t_i$ subdivides $T_i$ into three edge-disjoint paths $S^i_1,S^i_2,S^i_3$
(some of which might be trivial) so that $S^i_j$ contains the endvertex of $o_j$ (for $i=1,2$
and $j=1,2,3$).

Pick two distinct  $k,\ell\in\{1,2,3\}$ so that for $i=1,2$ at least two paths in 
$P_1,P_2,P_3$ the endvertex in $T_i$ is contained in $S^i_k\cup S^i_\ell=:R_i$. 
Let $\{m\}=\{1,2,3\}\sm \{k,\ell\}$.
Should now $P_j$ have its endvertex $p$ in $S^1_m-S^1_k- S^1_\ell$ 
concatenate the subpath $pS^1_mt_1$ with $P_j$,
and proceed in a similar way in $T_2$. 
In this way we turn the edge-disjoint $T_1$--$T_2$-paths into edge-disjoint $R_1$--$R_2$-paths.
Now, we obtain the desired contradiction from~\eqref{no3}.
\end{proof}

Next, we state a simple lemma that, however, is the key to reducing the number of 
$\mathcal P$-odd edges.

\begin{lemma}\label{lem:flip}
Let $G$ be a graph with a bipartition $\mathcal P$.
Given an edge-cut $F$ of $G$ that contains more $\mathcal P$-odd edges than $\mathcal P$-even edges, one can  compute a bipartition $\mathcal P'$ of $G$ with less $\mathcal P'$-odd edges in polynomial time.
\end{lemma}

\begin{proof}
Let $F=E(X,Y)$ separate $X\subseteq V(G)$ 
from $Y\subseteq V(G)$ in $G$.
Then put $\mathcal P':=(A\triangle X,B\triangle X)$, and observe that every $\mathcal P$-odd
edge in $F$ becomes $\mathcal P'$-even, while the edges outside $F$ do not change.
\end{proof}

Putting together the lemmas presented so far, we arrive at the following procedure.

\begin{lemma}\label{triads}
There is a polynomial-time algorithm that takes as input a $2$-connected subcubic graph $G$, 
a bipartition $\mathcal P$ and three $\mathcal P$-odd edges $o_1,o_2,o_3$.
The algorithm: 
\begin{enumerate}[\rm (a)]
\item either correctly decides that $G$ contains
a skewed theta;
\item or  computes an edge cut $F$ that contains more
$\mathcal P$-odd edges than $\mathcal P$-even edges.
\end{enumerate}
\end{lemma}
\begin{proof}
We describe the algorithm in the course of this lemma.
We omit a detailed discussion about the runtime complexity as the steps of the 
algorithm rely  on basic operations or reduce to solving  min-cut/max-flow
problems.
\astart 
\item If $G_\mathcal P$ is not connected, choose a component $X$ of $G_\mathcal P$
and return $F=E(X,G-X)$.
\suspend{enumerate}
Since $G$ is $2$-connected, $F$ contains at least two $\mathcal P$-odd edges,
which is condition~(b).
Let us now assume that $G_{\mathcal P}$ is connected.
\aresume
\item Compute a spanning tree $T$ of $G_\mathcal P$ and determine the fundamental cycles 
$C_{o_1}, C_{o_2}, C_{o_3}$ of $o_1,o_2,o_3$.
\item\label{exita} If any two of $C_{o_1}, C_{o_2}$ and $C_{o_3}$
are edge-disjoint, return \returnvalue{skewed theta}.
\suspend{enumerate}

The return value in line\aref{exita} is justified by Lemma~\ref{2oddcycs},
which means that we may assume the cycles $C_{o_1}, C_{o_2}, C_{o_3}$
to pairwise share an edge from now on.  
\aresume
\item\label{twotrees} If there is an edge $e$ shared by each of
$C_{o_1}, C_{o_2}, C_{o_3}$:
\begin{enumerate}[a.]
\item Let $T_1$ and $T_2$ be  
 the two components  of 
$\bigcup_{i=1}^3C_{o_i}-e$.
\item\label{pruning} Delete leaves from $T_1$ and $T_2$ until $T_1$ and $T_2$ have
the form of Lemma~\ref{lem:ThreeOddEdgeCut}.
\item Compute a smallest cut $F'=E_{G_\mathcal P}(X,Y)$ of $G_\mathcal P$ 
that separates $T_1$ from~$T_2$
\item\label{treeexit} If $|F'|\geq 3$, return \returnvalue{skewed theta}; otherwise return $F=E_G(X,Y)$.
\end{enumerate}
\suspend{enumerate}
Note that, for $i=1,2,3$, both components of $C_{o_i}-\{e,o_i\}$  contain
an endvertex of $o_i$, so that, 
after pruning,
$T_1$ and $T_2$ indeed
 conform with Lemma~\ref{lem:ThreeOddEdgeCut}.
Lemma~\ref{lem:ThreeOddEdgeCut} implies that $G$ contains
a skewed theta if $|F'|\geq 3$. Otherwise, $F$ contains at most
two $\mathcal P$-even edges and the three $\mathcal P$-odd edges $o_1,o_2,o_3$. 

Considering line\aref{twotrees}, 
we may from now 
on assume that there is no common edge of $C_{o_1},C_{o_2},C_{o_3}$.
Then
\begin{equation}\label{eq:FourPaths}
\begin{minipage}[c]{0.8\textwidth}\em
there
is a unique cycle $D$ in $\bigcup_{i=1}^3C_{o_i}$
that passes through each of $o_1,o_2,o_3$ and so that 
there is  a path in $G_\mathcal P$ between any 
two of the components of $D-\{o_1,o_2,o_3\}$
 that avoids the third.
\end{minipage}\ignorespacesafterend 
\end{equation} 
Indeed, each $C_{o_i}-o_i$ is a subpath of $T$ and families of 
subtrees of a tree
are known to have the Helly property, that is, if any two share
a vertex then there is also a common vertex to all.
Let $x$ be such a vertex.
Now, assume that $C_{o_1},C_{o_2},C_{o_3}$ do not have a common edge. 
Note that, for any $i \neq j$, $C_{o_i}$ and $C_{o_j}$ meet along a path. 
It follows that $C_{o_1}\cup C_{o_2}\cup C_{o_3}$
decomposes into a cycle $D$ that passes through all of $o_1,o_2,o_3$
and three internally disjoint $x$--$D$-paths that each end
in a different component of $D-\{o_1,o_2,o_3\}$.
Uniqueness of $D$ follows from the fact that $\bigcup_{i=1}^3C_{o_i}-\{o_1,o_2,o_3\}$
is a tree.
This proves~\eqref{eq:FourPaths}.

\aresume 
\item Determine the cycle $D$ in $\bigcup_{i=1}^3C_{o_i}$ that passes
through $o_1,o_2$ and $o_3$.
\suspend{enumerate}
Finding $D$ is easy, as this is done in the tree $\bigcup_{i=1}^3C_{o_i}-\{o_1,o_2,o_3\}$. 
(Alternatively, we may argue
that $E(D)$ is exactly the set of those edges in $\bigcup_{i=1}^3C_{o_i}$
that lie in only one of the cycles $C_{o_i}$.) Let $S_1,S_2,S_3$ be the 
three components of $D-\{o_1,o_2,o_3\}$.

\aresume 
\item \label{singlecut}Check whether there is a single edge $e'$ that separates $S_1$
from $S_2\cup S_3$ in $G_\mathcal P$. If yes, 
return $E_G(X,Y)$, where $X$ and $Y$ are the two components of $G_\mathcal P-e'$.
\suspend{enumerate}
Two of the edges $o_1,o_2,o_3$ are in the cut $E_G(X,Y)$,
while the only $\mathcal P$-even edge in it is $e'$.

\aresume 
\item Compute two 
edge-disjoint $S_1$--$(S_2\cup S_3)$-paths $P,Q$ in $G_\mathcal P$
so that one ends in $S_2$ and the other in $S_3$.
\suspend{enumerate}
Let us explain how $P$ and $Q$ can be computed. First, we use a standard algorithm 
to find two edge-disjoint $S_1$--$(S_2\cup S_3)$-paths $P,Q$ in $G_\mathcal P$;
these exist by Menger's theorem and line\aref{singlecut}. 
If already one ends in $S_2$ and the other in $S_3$, we use these. 
So, assume that $P$ and $Q$ both end in $S_2$, say. 
By~\eqref{eq:FourPaths}, we can find an $S_1$--$S_3$-path $R$ in $G_\mathcal P-S_2$.
If $R$ is disjoint from $P$ and $Q$, we replace $Q$ by $R$.
If not, we follow $R$ until we encounter for 
the last time a vertex of $P\cup Q$, where we see $R$ directed from $S_1$ to $S_3$. 
Let us say this last vertex $q$ is in $Q$. 
Then, we replace $Q$ by $QqR$. 

\aresume 
\item If $P$ and $Q$ are not crossing on $D$ then return \returnvalue{skewed theta}.
\item \label{gah}Otherwise, choose an edge $e''$  that separates the endvertices of $P$ and $Q$ in $S_1$
and apply lines~\ref{pruning}--\ref{treeexit} to
the two components $T_1$ and $T_2$ of $(D-\{o_1,o_2,o_3,e''\}) \cup P\cup Q$.
\end{enumerate}
If $P$ and $Q$ are not crossing then $D\cup P\cup Q$ contains two disjoint odd cycles, 
and thus $G$ contains a skewed theta, by Lemma~\ref{2oddcycs}. If, on the other hand, $P$ 
and $Q$ are crossing then each of the two components $T_1$ and $T_2$ as in
line\aref{gah}
is incident with an endvertex of each of $o_1,o_2,o_3$.
\end{proof}


We now prove that for line graphs $t$-perfection can be checked in polynomial-time.
\begin{proof}[Proof of Lemma~\ref{detectline}.]
Let $G$ be a given graph.
If $G$ has maximum degree at least~$4$, its line graph 
$L(G)$ is not $t$-perfect by Lemma~\ref{tpline}.
Otherwise, we apply the algorithm below to the blocks of $G$
to check whether $G$ contains a skewed theta.
Clearly, any skewed theta is completely contained in a block of $G$.

\astart
\item Set $\mathcal P:=(V(G),\emptyset)$.
\item While there are at least~$2$ distinct $\mathcal P$-odd edges, do the following:
\begin{enumerate}[a.]
\item Run the algorithm of Lemma~\ref{triads}.
\item If the algorithm returns a cut $F=E_G(X,Y)$ with more $\mathcal P$-odd
edges than $\mathcal P$-even edges, apply Lemma~\ref{lem:flip}.
\end{enumerate}
\item Apply Lemma~\ref{lem:LinkagesAlgo} to decide whether $G$ contains a skewed theta.
\end{enumerate}
The algorithm runs in polynomial-time, as the number of $\mathcal P$-odd edges decreases in each iteration of the while loop.

Correctness holds as Lemma~\ref{tpline} guarantees that $L(G)$ is $t$-perfect if and only if $G$ does not contain a skewed theta.
\end{proof}

\section{Proof of Lemma~\ref{lem:LinkagesAlgo}}\label{sec:nomorelinkages}

After having reduced the number of $\mathcal P$-odd edges, we are in 
this section in the situation that at most two remain. 
We note that, in this setting, checking for a skewed theta can be 
reduced to several applications of \textsc{$k$-Disjoint Paths} 
with a $k$ of at most~$5$. For any fixed $k$ a 
polynomial time algorithm is known to exist, see for instance
Kawarabayashi, Kobayashi and Reed~\cite{KKR12}. 
However, there does not seem to be a practical algorithm known if $k\geq 3$. 

We therefore give here an algorithm for Lemma~\ref{lem:LinkagesAlgo} 
that only relies on the solution of \textsc{$2$-Disjoint Paths}, 
for which several explicit algorithms are known that are independent 
of the heavy machinery of the graph minor project.
We start by treating the case when there is only one odd edge.

\begin{lemma}\label{lem:oneedge}
Let $G$ be a subcubic graph with a bipartition $\mathcal P$ such that there is at most one $\mathcal P$-odd edge.
Then it can be decided in polynomial time whether $G$ has a skewed theta.
\end{lemma}

\begin{proof}

If $G$ does not have any $\mathcal P$-odd edge then it cannot 
contain a skewed theta, and if $G$ is not $2$-connected then 
any skewed theta lies in the block that contains the 
$\mathcal P$-odd edge. Thus, we may 
assume that the input graph $G$ is $2$-connected and contains 
a unique $\mathcal P$-odd edge, $xy$ say.
Let $\mathcal P = (A,B)$, and let $x,y\in A$. 
\astart
\item \label{smallinst}If $|V(G)|\leq 3$, return \returnvalue{no skewed theta}. 
\suspend{enumerate}

We perform, if possible, one of two reductions in order to make 
the instance size smaller. 
If both $x$ and $y$ are of degree~$2$, then we add an edge between the neighbour $x'\neq y$
of $x$ and the neighbour $y'\neq x$ of $y$, and we delete $x,y$. See Figure~\ref{singlereduction}~(a)
for an illustration. (Observe that $x'\neq y'$, as $G$ is not a triangle.)
Denoting the resulting graph by $\tilde G$ and the induced bipartition by $\mathcal{\tilde P}$,
we note that the only $\mathcal{\tilde P}$-odd edge of $\tilde G$ is $x'y'$.
Moreover, $\tilde G$ has a skewed theta if and only if $G$ has a skewed theta.

\begin{figure}[ht]
\centering
\includegraphics[scale=0.7]{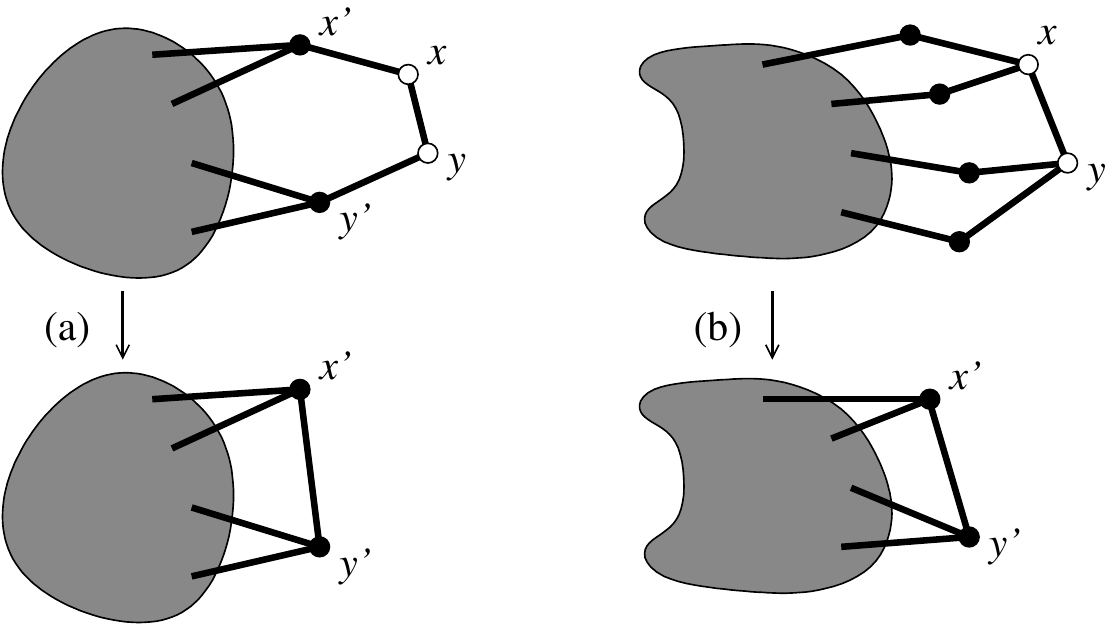}
\caption{Reduction $G\to\tilde G$}\label{singlereduction}
\end{figure}

In a similar way, we perform a reduction when 
both $x$ and $y$ have degree~$3$, and if each neighbour $u\notin\{x,y\}$
of~$x$ or of~$y$ has degree~$2$. 
Then we identify $x$ and $N(x)-y$ to a new vertex $x'$, 
and $y$ and $N(y)-x$ to a new vertex $y'$; see Figure~\ref{singlereduction}~(b).
Again, the resulting graph $\tilde G$ has a skewed theta precisely when $G$ has one;
and the only $\mathcal{\tilde P}$-odd edge is $x'y'$. 
\aresume
\item As long as possible, successively reduce $G$ to $\tilde G$.
\suspend{enumerate}
By exchanging $x$ and $y$, if necessary, we may therefore assume that
\begin{equation}\label{degcond}
\begin{minipage}[c]{0.8\textwidth}\em
$x$ has two neighbours $u,v\neq y$, and if $\deg_G(y)=3$ then $\deg_G(u)=3$ as well.
\end{minipage}\ignorespacesafterend 
\end{equation} 
The algorithm proceeds with
\aresume
\item\label{xnotbranchvx} Check whether $G$ contains a skewed theta, in which $x$ 
is a branch vertex.
\suspend{enumerate}
This is the case  if and only if there is a vertex $z \in B$ 
such that there are three paths between  $z$ and $\{y,u,v\}$ that have pairwise only $z$ in common.
Clearly, this can be checked for in polynomial time.

\aresume
\item If $G-x$ is $2$-connected return \returnvalue{skewed theta}.
\suspend{enumerate}
If $G-x$ is $2$-connected, then there is a cycle $C$ through $y$ and some other neighbour of $x$, say $u$.
Since $x \in A$ and thus $u \in B$, it follows that both paths from $y$ to $u$ in $C$ are of odd length.
Now, $C$ together with the path $yxu$ forms a skewed theta. 

So, we may assume that $G-x$ has cutvertices: Let their union with $N_G(x)$ 
be denoted with $S$.
Note that line\aref{xnotbranchvx} implies that
any skewed theta in $G$ has its two branch vertices in a common non-trivial block of $G-x$.
We prove:
\begin{equation}\label{eq:blockdegree}
\begin{minipage}[c]{0.8\textwidth}\em
every block of $G-x$ contains exactly two vertices of $S$, 
except for possibly one block, denoted by $X_*$, that contains three vertices of $S$.
\end{minipage}\ignorespacesafterend 
\end{equation} 
To prove the claim, consider the graph $H$ obtained from $G-x$ by adding three new
vertices $p_1,p_2,p_3$ each of which is precisely  adjacent to a distinct neighbour of $x$. 
Then the cutvertices of $H$ are exactly the vertices in $S$. 
Consider the block tree of $H$, that is, the graph defined on the blocks and cutvertices,
where a block $X$ and a cutvertex $w$ are adjacent if $w\in V(X)$.
Then, as $G$ is $2$-connected
every leaf in the block tree contains one of $p_1,p_2,p_3$. Thus, the block tree has at most
(in fact, precisely) three leaves, which directly gives Claim~\eqref{eq:blockdegree}.
\medskip

We use the following observation.
\begin{equation}\label{eq:non-trivialblocks}
\begin{minipage}[c]{0.8\textwidth}\em
if any non-trivial block $X$ of $G-x$ contains two vertices of $S$ 
in distinct classes of~$\mathcal P$ then $G$ contains a skewed theta.
\end{minipage}\ignorespacesafterend 
\end{equation} 
Suppose that~\eqref{eq:non-trivialblocks} is false.
Let $y'$ be a vertex of $S\cap V(X)$ for which there is a $y'$--$y$-path $P_y$
that is internally disjoint from $X$. Now, as~\eqref{eq:non-trivialblocks} is false
there is a vertex $z \in S\cap V(X)$ so that $y'$ and $z$ are not in the 
same bipartition class of $\mathcal P$. As $z\in S$, there is a path $P_z$ from $z$
to one of $u,v$, $u$ say,  that is internally disjoint from $X$. 
Let $C$ be a cycle in $X$ that contains $y'$ and $z$.
Then $C\cup P_y\cup P_z$ together with $uxy$ is a skewed theta with branch vertices $y',z$.
This proves~\eqref{eq:non-trivialblocks}.

\aresume
\item Compute the block decomposition of $G-x$, and check for~\eqref{eq:non-trivialblocks}.
\suspend{enumerate}

For every non-trivial block $X$ of $G-x$ we now construct a new graph $X'$,
so that
\begin{equation}\label{newblocksX}
\begin{minipage}[c]{0.8\textwidth}\em
 $G$ has a skewed theta both of whose branch vertices
are contained in $X$ if and only if $X'$ has a skewed theta. 
\end{minipage}\ignorespacesafterend 
\end{equation} 
Moreover, 
the bipartition $\mathcal P$ extends in a natural way to $X'$
so that there is precisely one $\mathcal P$-odd edge in $X'$.
The construction is sketched in Figure~\ref{blocksdecom}.

First consider a non-trivial block $X$ that contains exactly two vertices
of $S$, say $r$ and $s$. 
We observe that there is 
 an $r$--$s$-path $P$ in $G$ that is internally disjoint from $X$ 
and that passes through $xy$. As $r$ and $s$ are in the same class of $\mathcal P$,
the path $P$ has odd length. 
We set $X':=X+rs$. Clearly, any skewed theta of $X'$ contains $rs$. 
By replacing $rs$ with $P$, we then obtain a skewed theta of $G$. 
Conversely, a skewed theta of $G$ with both branch vertices in $X$ has a 
subpath from $r$ to $s$ that passes through $xy$. Substituting this subpath 
by $rs$ yields a skewed theta of $X'$. Thus,
we see that~\eqref{newblocksX} is satisfied.

\begin{figure}[ht]
\centering
\includegraphics[scale=0.7]{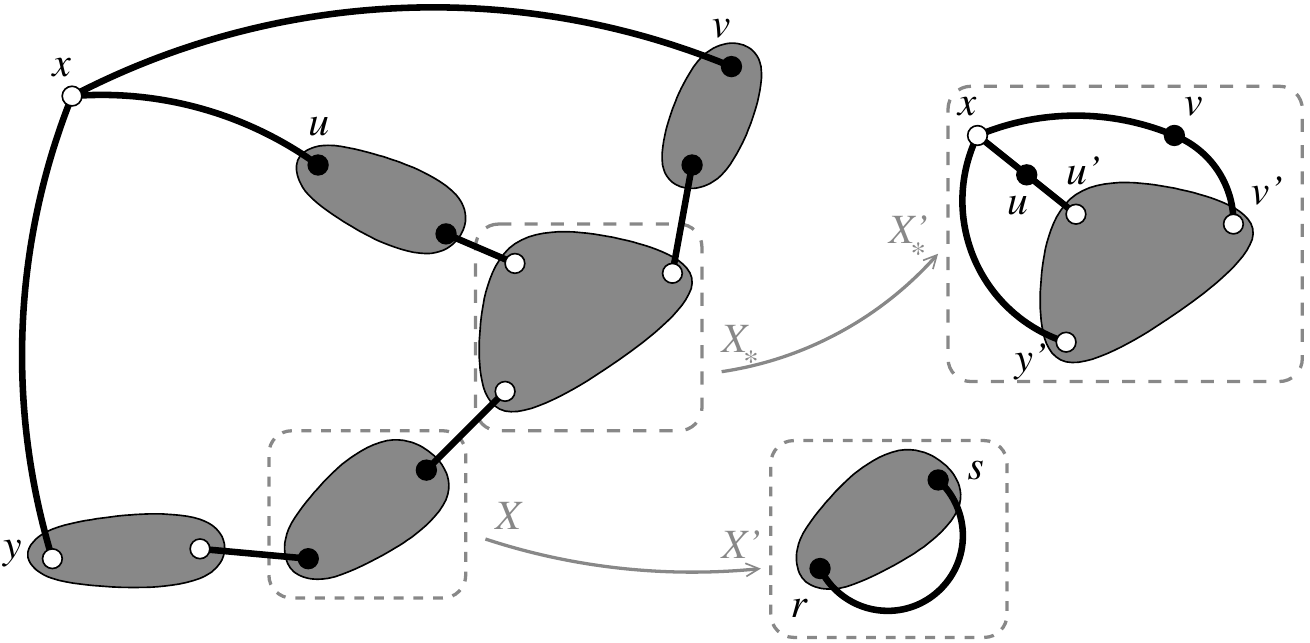}
\caption{The reduction of the blocks}\label{blocksdecom}
\end{figure}

Second, we treat the unique block $X_*$ containing three vertices 
of $S$, if there is such a block. Let the three vertices of $S$ in $X_*$
be $u',v',y'$, where the names are chosen such that there are 
disjoint paths $P_u,P_v,P_y$ linking $u$ to $u'$, $v$ to $v'$ and $y$ to $y'$,
and so that each of these paths is internally disjoint from $X_*$.

We claim that
\begin{equation}\label{3blockS}
\text{\em 
$\{u',v',y'\}\subseteq A$.
}
\end{equation}
Indeed, since we already checked for~\eqref{eq:non-trivialblocks}, 
either all of $\{u',v',y'\}$ are contained in $A$ or in $B$.
So, suppose that $\{u',v',y'\}\subseteq B$.
Now we find three internally disjoint paths between $y'$ and $x$,
which means that $x$ is a branch vertex of a skewed theta.
This, however, is impossible by~\eqref{xnotbranchvx}. 
To obtain the paths, start with the three paths $xyP_yy'$, $xuP_uu'$
and $xvP_vv'$, and extend the two latter paths by internally disjoint 
$\{u',v'\}$--$y'$-paths in $X_*$. These exists, since $X_*$ is a non-trivial block. 
This proves~\eqref{3blockS}.

We let now $X_*'$ be the graph obtained from $X_*$ by adding $x,u,v$
and the edges $xy',uu',vv'$. (Note, that $y=y'$ is possible, while 
$u,v\in B$ always implies $u\neq u'$ and $v\neq v'$.)
With this,~\eqref{newblocksX} is satisfied.

\aresume
\item Compute for every block $X$ of $G-x$ the graph $X'$
and apply line\aref{smallinst} to every $X'$ independently.
\suspend{enumerate}

In order to bound the total number of recursions called, we observe that
\begin{equation}\label{recurrenceeq}
\begin{minipage}[c]{0.8\textwidth}\em
$|V(X')|<|V(G)|$
for every non-trivial block $X$ of $G-x$,
and
$\sum_{X}|V(X')|\leq |V(G)|+2$, 
where the sum ranges over the  non-trivial blocks.
\end{minipage}\ignorespacesafterend 
\end{equation} 
Indeed, the second claim is immediate as $G$ is subcubic, which
is maintained throughout the algorithm, implies that no two non-trivial blocks of $G-x$
share a vertex. The only vertices that may appear in two $X_1',X_2'$ are $u,v$, and then 
only if one of $X_1,X_2$ is equal to $X_*$. The first claim needs only proof for $X_*$. 

So, suppose that $|V(X_*')|=|V(G)|$. Since in constructing $X_*'$ we add to $X_*$
the three vertices $x,u,v$, this is only possible if $y\in V(X_*)$, that is, if $y=y'$.
Then, since $X_*$ is a non-trivial block but $y$ is adjacent to $x\notin V(X_*)$, 
we deduce that $\deg_G(y)=3$, which by~\eqref{degcond} gives $\deg_G(u)=3$ as well. 
If $u$ had two of its neighbours in $X_*$ then $u$ itself would be contained in $X_*$, 
which is impossible as then $u=u'\in A$, by~\eqref{3blockS}, 
but $x\in A$ implies $u\in B$. Thus, $u$ has
besides $x$ a second neighbour outside $X_*$, which then also lies outside $X_*'$.
This shows that $|V(X_*')|<|V(G)|$.
\medskip

Using a standard analysis of the recurrence relation\footnote{See for example the textbook by Cormen et al.~\cite[Ch.~I.4]{CLRS09}.} 
given by~\eqref{recurrenceeq}, 
we get that the total number of recursions is $\mathcal O(|V(G)|^2)$.
Indeed, the input graph is split up into, essentially, disjoint parts, each of which is properly smaller than~$G$.
\end{proof}

\begin{lemma}\label{lem:2oddedges}
There is a polynomial-time algorithm that, 
given a $2$-connected subcubic graph $G$
and given a bipartition $\mathcal P$ of its 
vertex set so that there are exactly two $\mathcal P$-odd edges $o_1,o_2$,
either
\begin{enumerate}[\rm (a)]
\item\label{2odda} decides correctly that $G$ has a skewed theta;
\item\label{2oddb} or computes a minimal cut $F$ containing $o_1,o_2$ and at most
two other edges.
\end{enumerate}
\end{lemma}
\begin{proof}
Since the algorithm below can be reduced to min cut/max-flow problems,
it clearly can be implemented to run in polynomial time. 
\astart
\item\label{P1P2} Compute two disjoint paths $P_1,P_2$, each of which linking
one endvertex of $o_1$ to one endvertex of $o_2$.
\item Compute a minimal cut $F$ in $G$ separating $P_1$ from $P_2$.
\item If $|F|\leq 4$, return $F$.
\item If $|F|\geq 5$, return \returnvalue{skewed theta}.
\end{enumerate}
For the proof of correctness, observe first that paths $P_1,P_2$ as in
line\aref{P1P2} exists as $G$ is $2$-connected. Moreover, $F$ contains $o_1,o_2$.
Thus, if $|F|\leq 4$ we have indeed  outcome~\eqref{2odda}. So,
suppose that $|F|=5$, which implies that there
is a set $\mathcal Q$ of three
edge-disjoint $P_1$--$P_2$-paths. From $\Delta(G)\leq 3$
it follows that the paths  in $\mathcal Q$ are, in fact, pairwise disjoint. 
Now, if any two of them are not crossing on the cycle $C:=P_1\cup P_2+o_1+o_2$
then $G$ contains two odd disjoint cycles and therefore a skewed theta, 
by Lemma~\ref{2oddcycs}. So, we may assume that any two of them cross on $C$.

We observe
that two of the paths in $\mathcal Q$, let us say $R,S$, have their endvertices on
$P_1$ in the same class of $\mathcal P$. Let the endvertices of $R$ be $r_1$
and $r_2$, and $s_1$ and $s_2$ those of $S$, where $r_1$ and $s_1$ lie in $P_1$.
Then deletion of the internal vertices of $r_2P_2s_2$ from $C\cup R\cup S$
yields a skewed theta with $r_1,s_1$ as branch vertices.
\end{proof}

\begin{lemma}\label{lem:twoedges}
There is a polynomial-time algorithm that, 
given a subcubic graph $G$ with a bipartition $\mathcal P$ of its 
vertex set so that there are exactly two $\mathcal P$-odd edges $o_1,o_2$
and given a minimal cut $F$ containing $o_1,o_2$ and at most two other edges, decides whether $G$ has a skewed theta.
\end{lemma}
\begin{proof}
We first reduce to the relevant blocks of the graph.
\astart
\item\label{twoblocks} If the $\mathcal P$-odd edges are in separate blocks, 
apply Lemma~\ref{lem:oneedge} to both blocks in order to decide whether $G$ contains a skewed theta.
\item\label{singleblock} If both edges are in a single block, say $B$, set $G:=B$ and continue.
\suspend{enumerate}
So we may assume that $G$ is 2-connected.
Next we try to find an even smaller cut containing $o_1,o_2$.
\aresume
\item\label{smallercut} Check whether there is an edge $e$, so that $\{o_1,o_2,e\}$ is 
a cut, and if yes, apply Lemma~\ref{lem:flip} to $F'=\{o_1,o_2,e\}$
and then Lemma~\ref{lem:oneedge} in order to decide whether $G$ contains a skewed theta.
\suspend{enumerate}
We allow here  that $e\in\{o_1,o_2\}$.
From line\aref{smallercut} follows, in particular, that $|F|=4$, say $F=\{o_1,o_2,e_1,e_2\}$.
\aresume
  \item\label{simpletheta1} Apply Lemma~\ref{lem:oneedge} to $G-o_1$ and to $G-o_2$.
  \item\label{simpletheta2} Apply Lemma~\ref{lem:flip} and then Lemma~\ref{lem:oneedge} to $G-e_1$ 
and to $G-e_2$.
\suspend{enumerate}
Since $F$ is minimal, there are two components $C_1$, $C_2$ of $G-F$.
After lines\aref{simpletheta1} and\aref{simpletheta2}, we are sure that
\begin{equation}\label{Fskewedthetas}
\begin{minipage}[c]{0.8\textwidth}\em
any skewed theta of $G$ contains every edge of $F$. In particular, 
both branch vertices either lie in $C_1$ or in $C_2$. 
\end{minipage}\ignorespacesafterend 
\end{equation} 

We know from Lemma~\ref{2oddcycs} that two disjoint odd cycles imply the presence 
of a skewed theta. As we have only two $\mathcal P$-odd
edges, the problem reduces here to the \textsc{$2$-Disjoint Paths} problem, which may
be handled, for example, with the algorithm of Tholey~\cite{Tho06}.
\aresume
\item\label{exclude2oddcycs} Check whether $G$ contains two disjoint odd cycles, and if 
yes return \returnvalue{skewed theta}.
\suspend{enumerate}
Next, we prove that
\begin{equation}\label{linkingpaths}
\begin{minipage}[c]{0.8\textwidth}\em
for $i=1,2$, in $C_i$ there are internally disjoint paths $P_i=x_i\ldots u_i$ and $Q_i=y_i\ldots v_i$,
where $x_i,y_i$ are distinct endvertices of $o_1,o_2$ and $u_i,v_i$ are distinct endvertices of $e_1,e_2$.
\end{minipage}\ignorespacesafterend 
\end{equation} 
Indeed, suppose that there are no such paths in $C_1$, say. As $G$ is subcubic,
there are then also no two such paths that are merely edge-disjoint rather than vertex-disjoint.
Moreover, because $C_1$ is connected and $G$ subcubic,  no three edges
of  $o_1,o_2,e_1,e_2$ can have the same endvertex.
Thus there is an edge $e$ that separates in $C_1$ the endvertices of $o_1,o_2$
from the endvertices of $e_1,e_2$. Consequently, $\{o_1,o_2,e\}$ is a cut of $G$,
which is a case we had already discarded in line\aref{smallercut}.

\aresume
\item Compute $P_i,Q_i$ as in~\eqref{linkingpaths}.
\suspend{enumerate}
As $G$ does not contain any two disjoint odd cycles, we may assume that
\[
o_1=x_1x_2,\, o_2=y_1y_2,\, e_1=v_1u_2\emtext{ and }e_2=u_1v_2.
\] 
See Figure~\ref{fig2oddedges} for these edges.
Using again the fact that $G$ does not possess any two disjoint odd cycles, 
we may deduce that
\begin{equation}\label{nosecondlinkage}
\begin{minipage}[c]{0.8\textwidth}\em
for $i=1,2$, there are no two disjoint paths in $C_i$ linking $y_i$ to $u_i$ and
$x_i$ to $v_i$.
\end{minipage}\ignorespacesafterend 
\end{equation}

In the remainder of the proof, we compute two subcubic graphs $G_1$ and $G_2$ 
such that
\begin{equation}\label{divide}\text{\em 
$G$ contains a skewed theta if and only if $G_1$ or $G_2$ does.
}\end{equation}
Moreover, the restriction of $\mathcal P$ to $V(G_i)$ 
gives a bipartition $\mathcal P_i$ of $G_i$ with two $\mathcal P_i$-odd edges.

We only describe the construction of $G_1$; $G_2$ is obtained by reversing 
the sides $C_1$ and $C_2$.
We define a path $P_2'$ that is used to replace the path $x_1x_2P_2u_2v_1$ in $G_1$. 
If $P_2$ has odd length, we set  $P'_2:=x_1v_1$.
By considering the bipartition classes of $\mathcal P$,
we may see that $x_1\neq v_1$ and that the resulting new edge $x_1v_1$ is a $\mathcal P_1$-odd edge.
On the other hand, if $P_2$ has even length we 
set $P'_2:=x_1x_2v_1$. 
Note that in both cases the path $P'_2$
has the same parity as the path $x_1x_2P_2u_2v_1$ in $G$.
We define $Q_2'$ analogously and set $G_1:=C_1\cup P'_2\cup Q'_2$.
\begin{figure}[ht]
\centering
\includegraphics[scale=0.7]{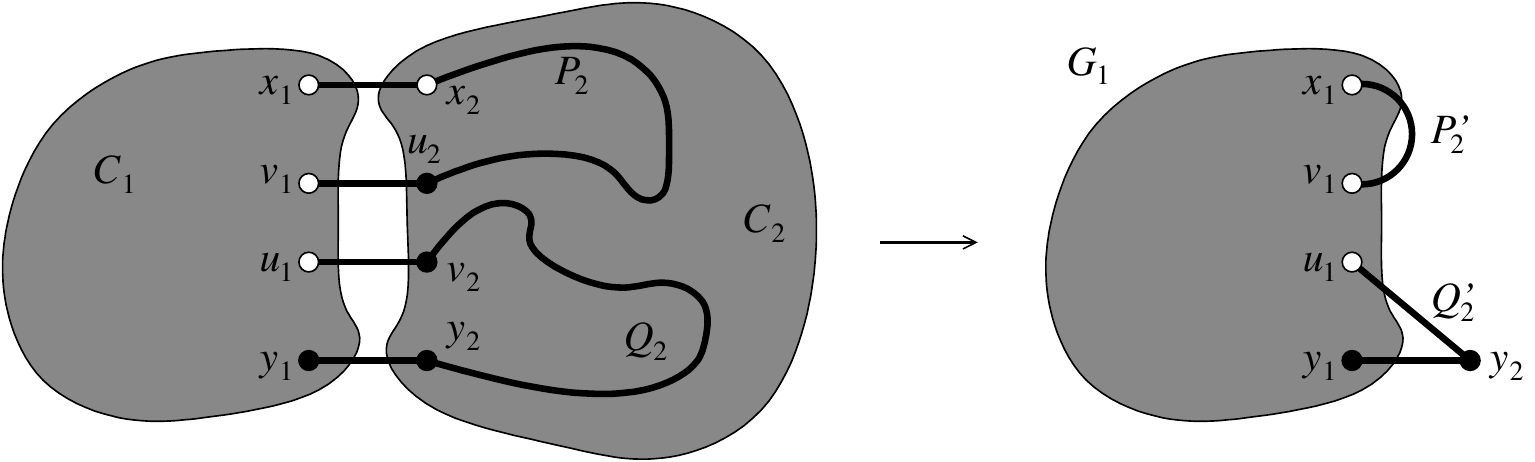}
\caption{Construction of $G_1$ if $P_2$ has odd length and $Q_2$ even length}\label{fig2oddedges}
\end{figure}

We note that for $i=1,2$
\begin{equation}\label{stuffgetssmaller}
|E(G_i)|<|E(G)|,\, \Delta(G_i)\leq 3\emtext{ and }
|E(G_1)|+|E(G_2)|\leq |E(G)|+4.
\end{equation}
While the last two inequalities should be clear, the first needs proof. As $G$ is subcubic
but $|F|=4$,
we deduce that $C_2$ has at least two vertices. As, on the other hand, $C_2$ is connected
we see that $C_2$ contains at least one edge. That edge, however, is missing in $G_1$,
which implies $|E(G_1)|<|E(G)|$. The proof for $G_2$ is the same.

To prove~\eqref{divide}, we first assume that $G$ contains a skewed theta $T$.
By~\eqref{Fskewedthetas}, $T$ has its two branch vertices $r,s$ either in $C_1$ or in $C_2$,
let us say that $r,s\in V(C_1)$. Moreover, each of the two odd paths
of $T$ between $r$ and $s$ passes through exactly one of $o_1,o_2$. 
Thus, the two odd paths contain subpaths $R,S\subseteq C_2$  
linking $\{x_2,y_2\}$ to $\{u_2,v_2\}$ in $C_2$.
From~\eqref{nosecondlinkage} it follows that one of $R$ and $S$, $R$ say, starts in $x_2$ and 
ends in $u_2$, while the other, $S$ in this case, connects $y_2$ to $v_2$.
Since the parity of the length of $R$ is determined by the classes of $\mathcal P$
that contain $x_2$ and $u_2$, it follows that the parity of the length of $R$
is the same as that of $P_2$, which is the same as that of $P'_2$.
Since the same reasoning holds for $S$ and $Q'_2$, we see that we obtain a skewed
theta of $G_1$ from $T$ by replacing $x_1x_2Ru_2v_1$ by $P'_2$ 
and $y_1y_2Sv_2u_1$ by $Q'_2$.

For the other direction, observe that any skewed theta of $G_1$ contains at least one
of $P'_2$ and $Q'_2$ (in fact both, but we do not need that observation). By replacing, if necessary,
$P'_2$ by $x_1x_2P_2u_2v_1$ and/or $Q'_2$ by $y_1y_2Q_2v_2u_1$,
we turn the skewed theta of $G_1$ into one of $G$.

\aresume
  \item Compute $G_1$ and $G_2$ and re-apply the algorithm to~$G_1$ and~$G_2$.
\end{enumerate}

Correctness of the algorithm follows from~\eqref{divide}.
It remains to analyse the running time of the algorithm.
Each line can be performed in polynomial time, 
so it suffices to bound the recursion.
Here,~\eqref{stuffgetssmaller} shows that the graph is split into two parts which are properly smaller and, essentially, disjoint.
A standard analysis of the recurrence relation shows that the total number of recursions called is $\mathcal O(|E(G)|^2)$.
\end{proof}

\begin{proof}[Proof of Lemma~\ref{lem:LinkagesAlgo}.]
The algorithm performs the following steps.
\astart
	\item If $G$ is not 2-connected, compute the blocks of $G$ and re-apply the algorithm to each block separately.
	\item If $G$ does not have any $\mathcal P$-odd edge, return \returnvalue{no skewed theta}.
	\item If $G$ has a single $\mathcal P$-odd edge, apply Lemma~\ref{lem:oneedge} to decide whether $G$ has a skewed theta.
	\item If $G$ has two $\mathcal P$-odd edges, apply Lemma~\ref{lem:2oddedges} to $G$, to compute the promised cut $F$. Then apply Lemma~\ref{lem:twoedges} to decide whether $G$ has a skewed theta. 
\end{enumerate}
Correctness and polynomial running time follow from the respective lemmas.
\end{proof}

\section{Claw-free graphs}\label{clawsec}

We now describe an algorithm that, given a claw-free graph $G$, decides 
in polynomial time whether $G$ is $t$-perfect or not. We present the algorithm 
in a number of steps over the course of this section. First, we use that 
we can already decide $t$-perfection for line graphs, and that we can 
detect whether a graph is a line graph efficiently:

\begin{theorem}[Roussopoulos~\cite{Rouss73}]
\label{sourcegraph}
It can be checked in linear time whether a given graph is a line graph.
Moreover, given a line graph $G$, a graph $H$ with $L(H)=G$ can be found in linear time.
\end{theorem}

Thus, the first step in the algorithm becomes:
\astart
	\item\label{compsource} Use Theorem~\ref{sourcegraph} to check whether $G$ is a line graph. If yes, 
compute $H$ with $L(H)=G$ and apply the algorithm 
of Lemma~\ref{detectline} to $H$. If no, proceed to the next line below.
\suspend{enumerate}
Next, we observe that we can assume the input graph to be $2$-connected. For this,
we say that 
a pair $(G_1,G_2)$ of proper induced subgraphs of a graph $G$ is 
a \emph{separation of $G$}, if $G=G_1\cup G_2$. The \emph{order}
of the separation is equal to $|V(G_1\cap G_2)|$.  

The following lemma may be deduced directly from the definition of $t$-perfection.
We only apply it to claw-free graphs, where it becomes a simple consequence of 
Theorem~\ref{thm:Char}.  
\begin{lemma}\label{completesep}
Let $(G_1,G_2)$ be a separation of a graph $G$ so that $G_1\cap G_2$ is complete.
Then $G$ is $t$-perfect if and only if $G_1$ and $G_2$ are $t$-perfect.  
\end{lemma}

\aresume
	\item Determine the blocks of $G$, and apply the rest of the algorithm 
to each block independently. Return \returnvalue{not $t$-perfect} if one of the blocks  
is not $t$-perfect; otherwise return \returnvalue{$t$-perfect}.
\suspend{enumerate}
Clearly, this step can be performed efficiently, and is, by Lemma~\ref{completesep},
correct. Thus, we may from now on assume $G$ to be $2$-connected.
Moreover, it is easy to see that $G$ is not $t$-perfect, if it contains a vertex of degree at least~$5$.
Indeed,
as $G$ is claw-free, the neighbourhood of any vertex $v$ of degree at least~$5$
always contains either a triangle or an induced $5$-cycle.
In the former case, the graph contains a $K_4$ and in the latter case a $5$-wheel
as induced subgraph. 

\aresume
	\item If $\Delta(G)\geq 5$ or if $G\in\{C^2_7,C^2_{10}\}$ return \returnvalue{not $t$-perfect}.
	\item\label{yesoutcome} If $G\in\{C^2_6-v_1v_6,C^2_7-v_7,C^2_{10}-v_{10}\}$ return \returnvalue{$t$-perfect}.
\suspend{enumerate}
That the three graphs in line\aref{yesoutcome} are $t$-perfect is proved in~\cite{BS12}. (In fact, $C_7^2$ and $C_{10}^2$ are {\em minimally $t$-imperfect},
that is, they are $t$-imperfect but every proper $t$-minor is $t$-perfect. The
graph $C^2_6-v_1v_6$ can be seen to be a $t$-minor of $C_{10}^2$.)

The remainder of the algorithm is based on the following lemma.
\begin{lemma}[Bruhn and Stein~\cite{BS12}]
\label{lem:outcomes}
Let $G$ be a $3$-connected claw-free graph of maximum degree at most~$4$.
If $G$ does not contain $K_4$ as $t$-minor 
then one of the following statements holds true:
\begin{enumerate}[\rm (a)]
\item $G$ is a line graph; or
\item $G\in\{C^2_6-v_1v_6,C^2_7-v_7,C^2_{10}-v_{10},C^2_7,C^2_{10}\}$.
\end{enumerate}
\end{lemma}

Thus,  we may assume that the input graph $G$ is $2$-connected
but not $3$-connected. That is, $G$ has a separation of order~$2$. 
\aresume
	\item \label{not3connline}If $G$ is $3$-connected, return \returnvalue{not $t$-perfect}. 
	\item Otherwise, find a separation $(G_1,G_2)$ of $G$ of order~$2$. 
Let $u,v$ be the two vertices in $G_1\cap G_2$.
\suspend{enumerate}
Line\aref{not3connline} is correct, as we had already excluded
that $G$ is a line graph, nor one of the exceptional graphs in~(b) of Lemma~\ref{lem:outcomes}.

To continue, we use a result that allows us to reduce the $t$-perfection of $G$
to the $t$-perfection of the two sides of the separation. 
For this, we write  $G_i\ide{u}{v}$ for 
the graph obtained from $G_i$ by identifying $u$ and $v$.
\begin{lemma}\label{lem:ClawFree2Connected}
Let $G$ be a $2$-connected claw-free graph of maximum degree at most~$4$. 
Assume $(G_1,G_2)$ to be a separation of $G$ with $V(G_1\cap G_2)=\{u,v\}$. Then:
\begin{enumerate}[\rm (i)]
\item\label{cond:BothParitiesNotPerfect}
If $G_1$ and $G_2$ each contain induced $u$--$v$-paths of both even and odd length, 
then $G$ is not $t$-perfect.
\end{enumerate}
Otherwise  $G$ is 
$t$-perfect if and only if $\tilde G_1$ and $\tilde G_2$ are $t$-perfect,
where 
\begin{enumerate}[\rm (i)]\setcounter{enumi}{1}
\item\label{cond:ono}
$\tilde G_1=G_1\ide{u}{v}$ and $\tilde G_2=G_2+uv$, 
if $G_1$ contains an odd induced $u$--$v$-path but $G_2$ does not;

\item\label{cond:nono}
$\tilde G_1=G_1$ and $\tilde G_2=G_2$, 
if neither of $G_1$ and $G_2$ contains an odd induced $u$--$v$-path;
 
\item\label{cond:ene}
$\tilde G_1=G_1+uv$ and $\tilde G_2=G_2\ide{u}{v}$,
if $G_1$ contains an even induced $u$--$v$-path but $G_2$ does not; and

\item\label{cond:nene}
$\tilde G_1=G_1$ and $\tilde G_2=G_2$,
if neither of $G_1$ and $G_2$ contains an even induced $u$--$v$-path.
\end{enumerate}
\end{lemma}

We defer the proof of Lemma~\ref{lem:ClawFree2Connected} to the next section.
We combine the lemma with the following algorithm: 
\begin{theorem}[van 't Hof, Kami\'nski and Paulusma~\cite{HKP12}]\label{thm:IndPathClawFree}
Given a claw-free graph $G$ and $u,v \in V(G)$, it can be decided in polynomial time whether there is an induced $u$--$v$-path of even (or of odd) length.
\end{theorem}

With this, our algorithm continues as follows:
\aresume
	\item Use Theorem~\ref{thm:IndPathClawFree} to 
determine the  parities
of induced $u$--$v$-paths in $G_1$ and in $G_2$. 
	\item If $G_1$ and $G_2$ each contain induced $u$--$v$-paths of both even and odd length,
return \returnvalue{not $t$-perfect}. 
	\item \label{lastline}Otherwise, choose $\tilde G_1$ and $\tilde G_2$ as in Lemma~\ref{lem:ClawFree2Connected},
and apply line\aref{compsource} to $\tilde G_1$ and to $\tilde G_2$ independently. 
Return \returnvalue{$t$-perfect} if both are $t$-perfect, and \returnvalue{not $t$-perfect} otherwise.
\end{enumerate}

We can finally complete the proof of our main result, that $t$-perfection can be checked 
for in polynomial time if the input is restricted to claw-free graphs.
\begin{proof}[Proof of Theorem~\ref{thm:detect}]
We have already seen that the algorithm described in the course of this section
is correct. Moreover, as each single line is executed in polynomial time, 
we only need to bound the number of times each line is executed.
For this, observe that every time there is a branching in line\aref{lastline},
the graph $\tilde G_1$ contains a vertex 
of $G$ that does not lie in $\tilde G_2$ and vice versa. Again, standard analysis 
of the recurrence yields that the number of iterations is bounded by $\mathcal O(|V(G)|^2)$.
\end{proof}

\section{Proof of Lemma~\ref{lem:ClawFree2Connected}}\label{sec:ClawLemma}

All that remains is Lemma~\ref{lem:ClawFree2Connected}. 
The first step in its proof consists of 
the observation that $t$-perfection in a claw-free graph depends
essentially only on the existence of $K_4$ as a $t$-minor.

\begin{lemma}\label{onlyK4counts}
A connected claw-free graph $G$ is $t$-perfect if and only if
\begin{enumerate}[\rm (i)]
\item $\Delta(G)\leq 4$; 
\item $G\neq C_7^2$ and $G\neq C^2_{10}$; and
\item $G$ does not contain $K_4$ as a $t$-minor.
\end{enumerate}
\end{lemma}
\begin{proof}
We had already seen above that a $t$-perfect claw-free graph 
has maximum degree at most~$4$. Thus, the forward direction is
obvious. For the other direction
assume $G$ to satisfy (i)--(iii) but suppose that $G$ is $t$-imperfect. 
By Theorem~\ref{thm:Char} and~(iii), $G$ contains $W_5$, $C_7^2$ or $C_{10}^2$ as
a proper $t$-minor. 

As $\Delta(G) \le 4$ and since $G$ is connected, 
neither of $W_5$, $C_7^2$ or $C_{10}^2$ appears as  induced subgraph in $G$.
Thus, $G$ has a $t$-minor $H$ so that a single $t$-contraction in $H$
results in $W_5$, $C_7^2$ or $C_{10}^2$. We choose $H$ to have 
a minimum number of vertices. 

We first note that, by~\eqref{keepscf}, the $t$-minor $H$ is still claw-free.
Moreover, we deduce that $\Delta(H)\leq 4$. 
Indeed, suppose that $\Delta(H) \ge 5$.
As $G$ does not contain $K_4$ as a $t$-minor, the same holds for $H$.
In particular, no neighbourhood of any vertex of degree $\Delta(H)$
contains a triangle.
So, it must contain $C_5$ as induced subgraph.
As no $t$-contraction transforms $W_5$ into $W_5$, $C_7^2$ or $C_{10}^2$,
this means in particular that
$H$ contains $W_5$ as a proper induced subgraph, 
which in turn implies that $H$ was not minimum.
\medskip

Let us first consider the case when a single $t$-contraction of $H$ yields $C_7^2$. 
Since $H$ is claw-free, the $t$-contraction is performed at a vertex $v''_1$ with 
exactly two neighbours denoted with $v_1'$ and $v'''_1$. We may assume
that the resulting new vertex of the $t$-contraction is $v_1$ of $C_7^2$; 
see Figure~\ref{fig:K4counts}.

      \begin{figure}[ht]
      \centering
      \input{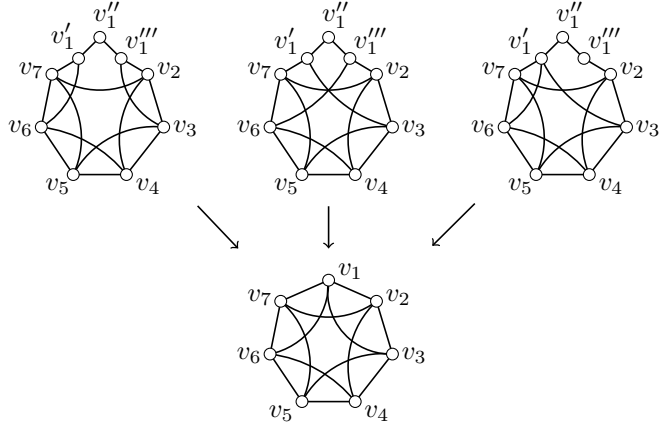}\vspace*{-0.7cm}
      \caption{Examples of single $t$-contractions that yield $C^2_7$}\label{fig:K4counts}
      \end{figure}

Now, as $v_1$ is adjacent to $v_2,v_3,v_6,v_7$, it follows that $N_H(v_1')\cup N_H(v_1''')=\{v_2,v_3,v_6,v_7\}$.
However, $v_1'$ cannot have two non-adjacent neighbours~$v_i,v_j$ among $v_2,v_3,v_6,v_7$,
as that would result in a claw on $v_i,v_j,v_1''$ with centre~$v_1'$. 
As the same holds 
for $v_1'''$, it follows that one of $v_1'$ and $v_1'''$ is adjacent to precisely $v_2,v_3$
while the other has exactly $v_6,v_7$ as neighbours among $v_2,v_3,v_6,v_7$. If, however, $N_H(v_1')=\{v_1'',v_6,v_7\}$
then $\{v_7,v_1',v_2,v_5\}$ induces a claw in $H$, which is impossible.

The case that $H$ can be $t$-contracted to $C_{10}^2$ is similar, 
so we skip to the case when $H$ contains $W_5$ as a $t$-contraction.

Let $v$ be the vertex at which the $t$-contraction is performed, let $u,w$ be its 
two neighbours in $H$, and let $x$ be the resulting vertex in $W_5$, which needs
to be the degree-$5$ vertex as $\Delta(H)\leq 4$.
Then, one of $u,w$, let us say~$u$, has at least three neighbours other than~$v$.
Since $H-\{u,v,w\}$ is a $5$-cycle, 
it follows that $u$ has at least two non-adjacent neighbours $y,z$ in 
$H- \{u,v,w\}$.
But then $\{u,v,y,z\}$ induces a claw in $H$, a contradiction.
This completes the proof.
\end{proof}

In general, it is not entirely straightforward to describe
the graphs from which $K_4$ can be obtained solely by $t$-contractions. 
For instance, Figure~\ref{fig:skewed} shows two quite different 
graphs that both $t$-contract to $K_4$. In claw-free graphs, 
in contrast, there is only one such type of graph.

      \begin{figure}[ht]
      \centering
      \includegraphics[scale=0.7]{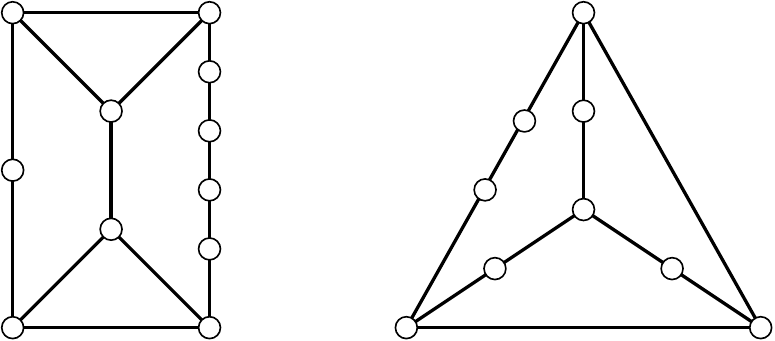}
      \caption{Two graphs that $t$-contract to $K_4$}\label{fig:skewed}
      \end{figure}

A \emph{skewed prism} (of a graph $G$)
is an induced subgraph of $G$ that consists of two triangles, say $x_1,x_2,x_3$ and $y_1,y_2,y_3$,
together with three vertex-disjoint induced paths $P_1$, $P_2$, and $P_3$,
each of which has one endvertex in $x_1,x_2,x_3$ and the other in $y_1,y_2,y_3$.
Moreover, we require the paths $P_1$ and $P_2$ to have even length, while $P_3$ has odd length.
(We allow $P_1$ and $P_2$ to have length~$0$.)
As an illustration, note that the graph on the left in 
 Figure~\ref{fig:skewed} is a skewed prism but the one on the right is not
(and it contains a claw).

Let us stress the fact that, in contrast to the skewed thetas treated in Section~\ref{sec:linegraphs},
skewed prisms are \emph{induced} subgraphs. 
Moreover, a skewed prism has two of its linking paths
even and one odd, while for a skewed theta it is the opposite: two odd, one even.
While this may create some confusion, we think that the name is nevertheless justified
by the clear connection of skewed thetas and prisms:
Indeed, the line graph of a skewed theta is a skewed prism, and moreover,
a graph $G$ contains a skewed theta if and only if its line graph $L(G)$ contains a skewed prism.

\begin{lemma}\label{K4sface}
A claw-free graph $G$ contains $K_4$ as a $t$-minor if and only if it contains a skewed prism.
\end{lemma}

\begin{proof}
By successively $t$-contracting vertices of degree~$2$,  one 
obtains  from any skewed prism a $K_4$.
Thus, if $G$ contains a skewed prism, it contains $K_4$ as $t$-minor.

For the other direction, let $H$ be a minimal induced subgraph of $G$ that can be $t$-contracted to $K_4$.
Suppose that $H$ is not a skewed prism.

Let $H_0, H_1, \ldots , H_k$ be a series of graphs with $H_0=H$ and $H_k \cong K_4$ such that $H_{i+1}$ is obtained from $H_i$ by a single $t$-contraction, for $i=0,\ldots,k-1$.
Note that, as $H$ is minimal, no proper induced subgraph of $H_i$ contains $K_4$ as $t$-minor, for all $i = 0, \ldots, k$.

As $H_k \cong K_4$ is a skewed prism, there is an index $i \le k-1$ such that $H_i$ is not a skewed prism but $H_{i+1}$ is.
Let $H_{i+1}$ consist of the two triangles
$x_1,x_2,x_3$ and $y_1,y_2,y_3$ and the disjoint $x_i$--$y_i$-paths $P_i$, for $i=1,2,3$,
so that $P_1,P_2$ have even length, while $P_3$ has odd length.
Assume that the $t$-contraction occurs at a vertex $v$ of $H_i$, which then identifies
its two neighbours $u,w$ to a new vertex $x$ of $H_{i+1}$.

We first observe that the neighbourhoods of $u$ and $w$ in $H_i$ are incomparable:
if, for example, $N_{H_i}(u) \subseteq N_{H_i}(w)$, 
then $H_{i+1} \cong H_i - \{u,v\}$, in contradiction to our observation that no proper induced subgraph of $H_i$ contains $K_4$ as $t$-minor.
Similarly, $|N_{H_i}(u)|,|N_{H_i}(w)| \ge 2$.

Let us discuss the case that $|N_{H_i}(u)|,|N_{H_i}(w)| \ge 3$.
Since $H_i$ is claw-free, both $N_{H_i}(u) \setminus \{v\}$ and $N_{H_i}(w) \setminus \{v\}$ are cliques.
This gives $|N_{H_i}(u)|,|N_{H_i}(w)| = 3$, since $H_i$ is, by minimality,  $K_4$-free.
As the neighbourhoods of $u$ and $w$ are incomparable, 
the  new vertex $x$ of $H_{i+1}$
is contained in two distinct triangles. 
Since the only two triangles in $H_{i+1}$ are $x_1,x_2,x_3$ and $y_1,y_2,y_3$,
we may assume that
$x = x_1 = y_1$ in $H_{i+1}$.
But then either $x_1=u$ and $y_1=w$ or $x_1=w$ and $y_1=u$ in $H_i$, which means that $H_i$ is a skewed prism (with $P_1=x_1vy_1$), a contradiction. 

The other cases are handled in a similar manner.
\end{proof}

Let $u,v$ be two distinct vertices in a graph $G$. 
A \emph{$u$--$v$-linked obstruction} is an induced subgraph of $G$ that consists of four vertex-disjoint induced paths $R$, $S$, $X$, and $Y$, so that the endvertices of $R$ are $u,r$, those 
of $S$ are $v,s$, and we write $x_1,x_2$ and $y_1,y_2$ for the endvertices of $X$ and $Y$, respectively.
The paths are required to satisfy the following conditions:
\begin{itemize}\itemsep1pt \parskip0pt 
	\item The vertices $r,x_1,y_1$ and $s,x_2,y_2$ form  triangles in $G$.
The edges of the two triangles are the only edges between $R$, $S$, $X$, and $Y$.
	\item The path $X$ has even length (where we allow length~$0$).
\end{itemize}

      \begin{figure}[ht]
      \centering
      \includegraphics[scale=1]{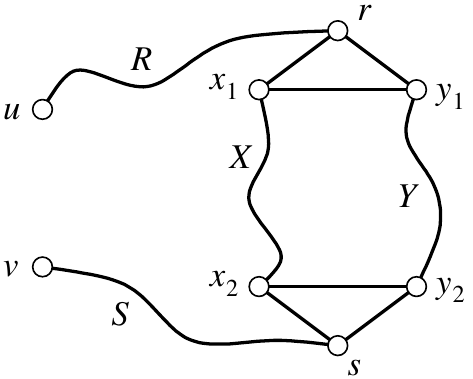}
      \caption{A $u$--$v$-linked obstruction}\label{fig:obstruction}
      \end{figure}

The following observation shows why $u$--$v$-linked obstructions are important:

\begin{lemma}\label{lem:Obstructions}
Let $(G_1,G_2)$ be a separation of a graph $G$ with $V(G_1\cap G_2)=\{u,v\}$.
If $G_1$ contains a $u$--$v$-linked obstruction and 
$G_2$ has two induced $u$--$v$-paths of distinct parity, then $G$ contains $K_4$ as $t$-minor.
\end{lemma}

\begin{proof}
Let $H$ be a $u$--$v$-linked obstruction in $G_1$
with paths $R,S,X,Y$.

First, let $Y$ have even length. By assumption, there is an induced $u$--$v$-path in $G_2$ 
such that the length of the induced path $rRuPvSs$ is odd.
Then, by $t$-contracting the vertices of degree~$2$ of $H\cup P$ we arrive at $K_4$.

Second, assume $Y$ to be an odd path, and choose $Q$ as an induced $u$--$v$-path in $G_2$ 
such that the induced path $rRuQvSs$ has even length.
Again, $H\cup Q$ can be $t$-contracted to $K_4$.
\end{proof}

Let us now prove that $u$--$v$-linked obstructions appear when induced $u$--$v$-paths of mixed parity are present:

\begin{lemma}\label{lem:MixedParitiesLink}
Let $G$ be a claw- and $K_4$-free graph with $\Delta(G)\leq 4$. 
Let furthermore $G$ be $2$-connected, and let
$(G_1,G_2)$ be a separation of $G$ with $V(G_1\cap G_2)=\{u,v\}$.
If there are two induced $u$--$v$-paths in $G_1$ of 
distinct parity, then $G_1$ contains a $u$--$v$-linked obstruction.
\end{lemma}

\begin{proof}
Let $P$ and $Q$ be two induced $u$--$v$-paths, where $P$ has even length and $Q$ odd length.
In particular, $uv \notin E(G)$.
We, furthermore, choose $P$ and $Q$ such that $|V(P) \cup V(Q)|$ is minimum  among all such pairs of paths.
Let $P=p_1\ldots p_r$ and $Q=q_1\ldots q_s$, where $u=p_1=q_1$ and $v=p_r=q_s$. 

Let us first observe:
\begin{equation}\label{morenbhs}
\begin{minipage}[c]{0.8\textwidth}\em
any $z\in V(G_1-Q)$ that has a neighbour $q\in V(Q)$ is also adjacent to 
one of the neighbours of $q$ in $Q$.
\end{minipage}\ignorespacesafterend 
\end{equation} 
Otherwise, there is a claw since
$q$ has three independent neighbours: $z$ and its two neighbours in $Q$
(if $q=u$ or $q=v$ pick a neighbour of $q$ in $G_2$ instead -- such a neighbour 
exists as $G$ is assumed to be $2$-connected).  

\medskip

We now assume that there is a vertex $x$ of $P$ that has at least three neighbours in $Q$.
In particular, $x$ does not belong to $Q$.

If $x$ has exactly three neighbours in $Q$
we deduce from~\eqref{morenbhs} that
they appear consecutively on $Q$, that is, the neighbours are $q_iq_{i+1}q_{i+2}$ 
for some $i$.
In that case, $Q+x$ is a $u$--$v$-linked obstruction, 
where we choose $R = uQq_i$, $S = q_{i+2}Qv$, $X=\{x\}$ and $Y=\{q_{i+1}\}$.

If $x$ has more than three neighbours in $Q$, then it has exactly four as $\Delta(G)\leq 4$.
By~\eqref{morenbhs}, there is $i<j$ so that the neighbours are $q_i,q_{i+1},q_j,q_{j+1}$.
Again, we find that $Q+x$ is a $u$--$v$-linked obstruction:  
Set $R = uQq_i$, $S=q_{j+1}Qv$, $X=\{x\}$ and $Y=q_{i+1}Qq_j$.

By symmetry, we may thus assume that
\begin{equation}\label{twonbhs}
\emtext{
every vertex of $Q$ has at most two neighbours in $P$, and vice versa. 
}\end{equation}

Choose $i$ minimum such that $p_i \neq q_i$.
As $P,Q$ are  induced paths, this implies that  $p_i \notin V(Q)$, from which 
with~\eqref{morenbhs} follows that $p_i$ and $q_i$ are adjacent.
Since $P$ and $Q$ have the same endvertex, we may moreover
choose a minimum $j \geq i$ so that $p_{j+1} \in V(Q)$. 

We claim that
\begin{equation}\label{claim:NoNeighbourInQ}
\mbox{\emph{no vertex of the path $p_{i+1}Pp_{j-1}$ has a neighbour in $Q$.}}
\end{equation}

In order to prove the claim, 
suppose by way of contradiction that there is a minimum $\ell\in\{i+1,\ldots,j-1\}$ 
so that $p_\ell$ has a neighbour $x$ in $Q$. 

Suppose that $p_{\ell-1}x' \in E(G)$
for some neighbour $x'\in V(Q)$ of $p_\ell$, which by the minimality of $\ell$
is only possible when $i+1=\ell$. Since $p_{i+1}$ is not a neighbour of $q_{i-1}=p_{i-1}$,
it follows that $x'\neq q_{i-1}$. Then $x'=q_i$, as $p_i$ cannot have three distinct 
neighbours $q_{i-1},q_i,x'$ in $Q$ by~\eqref{twonbhs}. But now $q_i$ has three neighbours
in $P$, namely $p_{i-1},p_i,p_{i+1}$, contradicting~\eqref{twonbhs}.

In particular, with $x$ in the role of $x'$, we obtain that $p_{\ell-1} x \notin E(G)$.
The choice of $j$ together with $x\in V(Q)$ implies that $p_{\ell+1} \neq x$,
as $\ell+1\leq j$. Thus, $x\notin V(P)$ and we deduce with~\eqref{morenbhs}
that $x$ is adjacent to $p_{\ell+1}$. Because also $p_\ell\notin V(Q)$
(by choice of~$j$), we obtain from~\eqref{morenbhs} that $p_\ell$
is adjacent to a neighbour $y$ of $x$ in $Q$. Again, $q_i\neq y$ as otherwise
$q_i$ had the three neighbours $p_{i-1},p_i,p_{\ell}$ in $P$, contradicting~\eqref{twonbhs}.
We apply~\eqref{morenbhs} again to see that $y$ is adjacent to either $p_{\ell-1}$ or to $p_{\ell+1}$. 
The former case, however, is impossible by the above observation that no neighbour $x'\in V(Q)$
of $p_\ell$ is adjacent to $p_{\ell-1}$.

Thus, $p_{\ell+1} y \in E(G)$, which means that 
$\{x,y,p_l,p_{l+1}\}$ induces a $K_4$ in $G$, a contradiction.
This proves~\eqref{claim:NoNeighbourInQ}.

\medskip
Let $p_{j+1}=q_k$, and observe that, as $p_j\notin V(Q)$ by minimality of $j$,
it follows from~\eqref{morenbhs} and~\eqref{twonbhs} that $p_j$ is adjacent
to $q_{k-1}$ or to $q_{k+1}$, but not to both.

We first consider the case that $p_jq_{k-1}\in E(G)$.
Suppose  that the lengths of the paths $p_iPp_j$ and $q_iQq_{k-1}$ have the same parity.
Then, we may replace in $Q$ the subpath $q_iQq_{k-1}$ by $p_iPp_j$. The obtained $u$--$v$-path 
$Q':=uQq_{i-1}Pp_jq_{k}Qv$ then has odd length, exactly as $Q$. 
Moreover, $Q'$ is induced by~\eqref{claim:NoNeighbourInQ}.
Since $|V(P) \cup V(Q')| < |V(P) \cup V(Q)|$ we obtain a contradiction to the 
choice of $P$ and $Q$.
Therefore, $p_iPp_j$ and $q_iQq_{k-1}$ have different parities.
But then the subgraph induced by  
$Q\cup p_iPp_j$ is a $u$--$v$-linked obstruction: We 
let $R = uQq_{i-1}$, $S = q_{k}Qv$, and for $X$ we choose the path 
among $p_iPp_j$ and $q_iQq_{k-1}$ of even  length, and for $Y$ the odd one.

If $p_j$ is adjacent to $p_{k+1}$ (and then not to $p_{k-1}$), 
we argue in a similar way in order to see that $p_iPp_j$ and $q_iQq_k$
have different parities.  Then, we may choose $R = uQq_{i-1}$, $S = q_{k+1}Qv$, and 
$X,Y$ as $p_iPp_j$ and $q_iQq_{k}$, depending on the parity.
\end{proof}

We can now prove our main lemma.

\begin{proof}[Proof of Lemma~\ref{lem:ClawFree2Connected}]
If the edge $uv$ is present in $G$, then every induced $u$--$v$-path in $G_1$
or in $G_2$ is odd (as the edge is the only induced path). Thus, we are 
in case~\eqref{cond:nene}, which reduces to Lemma~\ref{completesep}. Therefore, we may
assume from now on that $uv\notin E(G)$.

For~\eqref{cond:BothParitiesNotPerfect}, note that we may assume $G$ 
to be $K_4$-free, since $K_4$ is not $t$-perfect.
Thus, Lemma~\ref{lem:MixedParitiesLink} implies that $G_1$ contains a $u$--$v$-linked obstruction,
which means we find $K_4$ as a $t$-minor in $G$, 
by Lemma~\ref{lem:Obstructions}. Thus $G$ is not $t$-perfect.
\medskip

For the forward direction of~\eqref{cond:ono}--\eqref{cond:nene}, observe that 
the parity conditions guarantee that the respective $\tilde G_1,\tilde G_2$ are $t$-minors 
of $G$. Thus, $t$-perfection of $G$ also implies their $t$-perfection.
\medskip

For the back direction of~\eqref{cond:ono}--\eqref{cond:nene}, we assume
$G$ to be $t$-imperfect. 
Note that $G\notin\{C_{7}^2,C_{10}^2\}$ as both of the latter graphs
are $3$-connected but $G$ is not.
With Lemmas~\ref{onlyK4counts} and~\ref{K4sface}
we deduce that $G$ has a skewed prism $H$ consisting of two 
triangles $x_1,x_2,x_3$ and $y_1,y_2,y_3$ and of three 
linking paths $P_i=x_i\ldots y_i$ ($i=1,2,3$).

Let us examine how $H$ can be positioned with respect to the separation $(G_1,G_2)$.
There are three possibilities:
\begin{enumerate}[\rm (a)]
\item $H\cap G_1$ is empty or $H\cap G_2$ is empty;
\item $G$ contains $K_4$ as a subgraph; or
\item $H\cap G_1$ is a subpath of one of $P_1,P_2,P_3$, or that 
is the case for $H\cap G_2$.
\end{enumerate}
In order to prove that (a)--(c) covers every case, we may by symmetry assume
that $H\cap G_2$ contains the edge $x_1x_2$ of $H$. Now, we consider first the
case when $H\cap G_1$ is non-empty but devoid of edges.
In particular, that implies $H\subseteq G_2$.  
Let us assume that $u$ lies in $H\cap G_1$ (and possibly $v$, too). 
We observe that $u$ is adjacent to a vertex in $G_1$, as $G$ is $2$-connected. Thus,
the absence of claws implies that the neighbours of $u$ in $G_2$ 
are pairwise adjacent. 
One of the three linking paths $P_1,P_2,P_3$ of $H$ contains $u$, $P_1$ say.
We deduce that $P_1$ has to have length at most~$1$, 
as otherwise the two neighbours of $u$ in $P_1\subseteq G_2$ is adjacent
(if $u$ is an internal vertex) or
one of the triangle vertices $x_2,x_3,y_2,y_3$ is adjacent to an 
internal vertex of $P_1$ (if $u$ is an endvertex of $P_1$).
Now, whether $P_1$ has length~$0$ or~$1$, in both cases $u$ has three distinct neighbours among $x_1,x_2,x_3,y_1,y_2,y_3$. As those neighbours need to 
be pairwise adjacent, we have found $K_4$ as a subgraph of $G$.

It remains to consider the case when $H\cap G_1$ is non-empty and contains an
edge. Since any pair $x_i,y_j$ is connected by three internally disjoint
paths in $H$, we see that all of $x_1,x_2,x_3,y_1,y_2,y_3$
lie in $G_2$. Therefore, any edge of $H$ in $G_1$ is an edge of one 
of the linking paths $P_1,P_2,P_3$, and clearly of only one of them.
Thus, $H\cap G_1$ is a subpath of one of $P_1,P_2,P_3$.
This proves that (a)--(c) exhaust all possibilities.

\medskip
We now apply (a)--(c) to the back direction of~\eqref{cond:ono}.
If $H\cap G_1$ or $H\cap G_2$ is empty, then in particular $H$ 
is disjoint from $u,v$ and therefore, $H$ is still a skewed prism
of either $G\ide{u}{v}$ or of $G_2+uv$. By Lemma~\ref{K4sface}, 
one of the two is then $t$-imperfect.
If $G$ contains $K_4$ as a subgraph, then at most one of $u,v$
can lie in the $K_4$ as we assumed $uv\notin E(G)$. 
Consequently, $K_4$ is still a subgraph of one of $G\ide{u}{v}$ or $G_2+uv$.

It remains to consider option~(c).
If $H\cap G_1$ is a subpath of one of $P_1,P_2,P_3$, 
then the subpath needs to be of odd length, as 
every induced $u$--$v$-path through $G_2$ is assumed to be of even length. 
Replacing the odd path through $G_1$ by the edge $uv$, we obtain 
a skewed prism of 
of $G_2+uv$, as desired. 
If, on the other hand, $H\cap G_2$ is a subpath of one of $P_1,P_2,P_3$ 
then this subpath has even length by assumption. That means restricting $H$ to $G_1$ 
while identifying $u$ with $v$ yields a skewed prism of $G_1\ide{u}{v}$, 
and we are done.

\medskip
Next, we treat the back direction of~\eqref{cond:nono}. 
Observe that (a) and (b) imply that $H$ (or some $K_4$-subgraph) 
is completely contained in $G_1$ or in $G_2$, while~(c) is impossible. 
Indeed, if $H\cap G_1$ (or $H\cap G_2$) was a subpath of one of $P_1,P_2,P_3$,
then of necessarily even length, we would find an odd induced 
$u$--$v$-path in $H\cap G_2$ ($H\cap G_1$, respectively), contrary
to assumption.

\medskip The back directions of~\eqref{cond:ene} and~\eqref{cond:nene} are proved 
with similar arguments.
\end{proof}

\section{Discussion}

A key step for the recognition of claw-free $t$-perfect graphs is the insight
of
Lemmas~\ref{onlyK4counts} and~\ref{K4sface} 
that the problem reduces to the detection of skewed prisms. 

Skewed prisms are  induced subgraphs. 
As Fellows, Kratochvil, Middendorf and Pfeiffer~\cite{FKMP95} observed,
searching for a certain substructure often becomes substantially harder if 
one requires the substructure to be induced: finding the largest matching can be done
in polynomial time, but determining the size of the largest induced matching
is NP-complete. 

In the same way, checking for a non-induced prism (and without any parity constraints
on the paths) reduces to verifying whether between any two triangles
there are three disjoint paths, which clearly can be done in polynomial time. 
Checking whether a given graph contains an induced prism, however, is NP-complete -- 
this is a result of Maffray and Trotignon~\cite{MT05}. Interestingly, this changes
when the input graph is claw-free. Golovach, Paulusma and van Leeuwen~\cite{GPL} 
describe a polynomial-time algorithm for the induced variant of the $k$-\textsc{Disjoint Paths Problem} 
in claw-free graphs. By again considering any pair of triangles in a claw-free
graph, the algorithm may be used to detect prisms. Unfortunately, 
or rather fortunately for the purpose of this article, 
this is not 
enough to recognise $t$-perfection. For this, we need to detect \emph{skewed} prisms.
It is not clear whether the algorithm of
Golovach, Paulusma and van Leeuwen can be extended to incorporate parity constraints.

Kawarabayashi, Li and Reed~\cite{KLR10} give a polynomial-time algorithm 
to detect subgraphs arising from $K_4$ by subdividing its edges to odd paths. 
In our terminology, 
these are (non-induced) subgraphs that can be $t$-contracted to $K_4$.
Here the question arises whether one could develop and induced variant 
of their algorithm.

\bibliographystyle{amsplain}
\bibliography{tpbib}
 
\vfill
\small
\vskip2mm plus 1fill
\noindent
Version \today
\bigbreak

\noindent
Henning Bruhn
{\tt <henning.bruhn@uni-ulm.de>}\\
Universit\"at Ulm, Germany\\[3pt]
Oliver Schaudt
{\tt <schaudto@uni-koeln.de>}\\
Institut f\"ur Informatik\\
Universit\"at zu K\"oln\\
Weyertal 80\\
Germany

\end{document}